\documentclass[letterpaper,11pt]{article}

\usepackage{amsmath, amsthm, amssymb}
\usepackage[usenames, dvipsnames]{color}
\usepackage[normalem]{ulem} 
\usepackage{fullpage}
\usepackage[numbers, sort]{natbib}
\usepackage[noend]{algorithmic}
\usepackage{tikz}
\usepackage{enumerate}
\usepackage{hyperref}
\usepackage{xspace,color}
\usepackage{graphicx}
\usepackage{caption}
\usepackage{subcaption}
\usepackage{enumitem,linegoal}
\usepackage[linesnumbered,ruled,vlined]{algorithm2e}
\usepackage{comment}	
\usepackage{ctable}					

\newtheorem{fact}{Fact}[section]
\newtheorem{observation}{Observation}[section]
\usepackage{mathtools}
\usepackage[flushleft]{threeparttable}
\usepackage{verbatim}
\usepackage{float}
\usepackage{tabu}

\usepackage{bbm}

\usepackage{footnote}
\usepackage{authblk}
\usepackage{multirow}
\makesavenoteenv{tabular}
\makesavenoteenv{table}

\newtheorem{theorem}{Theorem}[section]

\newtheorem{lemma}[theorem]{Lemma}

\newtheorem{corollary}[theorem]{Corollary}

\newtheorem{definition}[theorem]{Definition}

\newtheorem{problem}[theorem]{Problem}

\renewcommand{\bar}[1]{\mkern 1mu\overline{\mkern-1mu#1\mkern-1mu}\mkern 1mu}
\newcommand{\nz}{\mathsf{nz}}
\newcommand{\map}{\mathsf{mp}}

\newcommand{\tildorder}{\widetilde O}

\newcommand{\matchx}[1]{\stackrel{\mathclap{\normalfont\scriptsize\mbox{$#1$}}}{\equiv}}
\newcommand{\matchq}{\matchx{?}}

\newcommand{\hash}{h}
\newcommand{\eunity}[2]{e^{-2\pi i #2 / #1}}
\newcommand{\unity}[2]{W_{#1}^{#2}}

\newcommand{\alf}{{\Sigma}}

\newcommand{\machines}{{\mathcal M}}
\newcommand{\memory}{{\mathcal S}}
\newcommand{\beforall}{\hspace{0.2cm}\operatorname{\forall}}

\newcommand{\str}[1]{{``#1''}}
\newcommand{\chr}[1]{{`#1'}}

\newcommand{\conv}{{\mathcal C}}

\newcommand{\qvec}{{\dagger}}
\newcommand{\rle}{{\circ}}

\AtBeginDocument{%
  \providecommand\BibTeX{{%
    \normalfont B\kern-0.5em{\scshape i\kern-0.25em b}\kern-0.8em\TeX}}}

\begin{document}

\title{String Matching with Wildcards in the Massively Parallel Computation Model}



%
%
%

\author[1]{MohammadTaghi Hajiaghayi }
\author[1]{Hamed Saleh}
\author[2]{Saeed Seddighin}
\author[3]{Xiaorui Sun}
\affil[1]{University of Maryland, College Park}
\affil[2]{Toyota Technological Institute at Chicago}
\affil[3]{University of Illinois at Chicago}
\date{}

\maketitle


\begin{abstract}
We study distributed algorithms for string matching problem in presence of wildcard characters. Given a string $T$ (a text), we look for all occurrences of another string $P$ (a pattern) as a substring of string $T$. Each wildcard character in the pattern matches a specific class of strings based on its type. String matching is one of the most fundamental problems in computer science, especially in the fields of bioinformatics and machine learning. Persistent effort has led to a variety of algorithms for the problem since 1960s. 

With rise of big data and the inevitable demand to solve  problems on huge data sets, there have been many attempts to adapt classic algorithms into the MPC framework to obtain further efficiency. MPC is a recent framework for parallel computation of big data, which is designed to capture the MapReduce-like algorithms. In this paper, we study the string matching problem using a set of tools translated to MPC model. We consider three types of wildcards in string matching:
\begin{itemize}
	\item \textbf{\chr{?} wildcard}: In this setting, the pattern is allowed to contain special \chr{?} characters  or don't cares that match any character of the text. String matching with don't cares could be solved by fast convolutions, and we give a constant round MPC algorithm for which by utilizing FFT in a constant number of MPC rounds.
	\item \textbf{\chr{+} wildcard}: \chr{+} wildcard is a special character that allows for arbitrary repetitions of a character. When the pattern contains \chr{+} wildcard characters, our algorithm runs in a constant number of MPC rounds by a reduction from subset matching problem.
	\item \textbf{\chr{*} wildcard}: \chr{*} is a special character that matches with any substring of the text. When \chr{*} is allowed in the pattern, we solve two special cases of the problem in logarithmic rounds.
\end{itemize}

\end{abstract}





\section{Introduction}
The string matching problem with wildcards, or \textit{pattern matching}, seeks to identify pieces of a text that adhere to a certain structure called the \textit{pattern}. Pattern matching is one of the most applied problems in computer science. Examples range from simple batch applications such as Awk, Sed, and Diff to very sophisticated applications such as anti-virus tools, database queries, web browsers, personal firewalls, search engines, social networks, etc. The astonishing growth of data on the internet as well as personal computers emboldens the need for fast and scalable pattern matching algorithms.

In theory too, pattern matching is a well-studied and central problem. The simplest variant of pattern matching, namely \textit{string matching}, dates back to 1960s. In this problem, two strings $T$ and $P$ are given as input and the goal is to find all substrings of $T$ that are identical to $P$. The celebrated algorithm of Knuth, Morris, and Pratt~\cite{knuth1977fast} (KMP) deterministically solves the problem in linear time. Since then, attention has been given to many variants and generalizations of pattern matching~\cite{hoffmann1982pattern,ramesh1992nonlinear,steyaert1983patterns,dufayard2005tree,knuth1977fast,lewenstein2014space,baron20125pm,brodie2006scalable,hosoya2001regular,mateus2005quantum,montanaro2017quantum,bennett1997strengths,ramesh2003string,ahmed2007distributed,holub2001distributed,clark2004scalable,chen2005boosting,agrawal2008efficient,galil1985optimal,vishkin1985optimal,breslauer1990optimal,breslauer1991lower}. Natural generalizations of string matching are when either the text or the pattern is a tree instead of a string~\cite{hoffmann1982pattern,ramesh1992nonlinear,steyaert1983patterns,dufayard2005tree} or when the pattern has a more sophisticated structure that allows for \chr{?}, \chr{+}, \chr{*}, or in general any regular expression~\cite{lewenstein2014space,baron20125pm,brodie2006scalable,hosoya2001regular}. Also, different computational systems have been considered in the literature: from sequential algorithms~\cite{hoffmann1982pattern,ramesh1992nonlinear,steyaert1983patterns,dufayard2005tree,knuth1977fast,lewenstein2014space,baron20125pm,brodie2006scalable,hosoya2001regular}, to quantum algorithms~\cite{mateus2005quantum,montanaro2017quantum,bennett1997strengths,ramesh2003string}, to distributed settings~\cite{ahmed2007distributed,holub2001distributed,clark2004scalable}, to the streaming setting~\cite{porat2009exact,chen2005boosting,agrawal2008efficient}, to PRAM~\cite{galil1985optimal,vishkin1985optimal,breslauer1990optimal,breslauer1991lower}, etc.

An obvious application of pattern matching is in anti-virus softwares. In this case, a malware is represented with a pattern and a code or data is assumed to be infected if it contains the pattern. In the simplest case, the pattern only consists of ascii characters. However, it happens in practice that malwares allow for slight modifications. That is, parts of the pattern code are subject to change. This can be captured by introducing wildcards to pattern matching. More precisely, each element of the pattern is either an ascii code or a special character \chr{?} which stands for a wildcard. The special character is allowed to match with any character of the text. 

Indeed the desirable property of the \chr{?} case is that the length of the pattern is always fixed. However, one may even go beyond this setting and consider the cases where the pattern may match to pieces of the text with variant lengths. Two classic ways to incorporate this into the model is to consider two special characters \chr{+} and \chr{*}. The former allows for arbitrary repetitions of a single character and the latter allows for arbitrary repetitions of any combination of characters. For example, as an application of pattern matching in bioinformatics, we might be looking for a set of gene patterns in a DNA sequence. Obviously, these pattern are not necessarily located consecutively in the DNA sequence, and one might utilize \chr{*} wildcard to address this problem.

In practice, these problems are formulated around huge data sets. For instance, a human DNA encompasses roughly a Gigabyte of information, and an anti-virus scans Gigabytes (if not Petabyets) of data on a daily basis. Thus, the underlying algorithm has to be scalable, fast, and memory efficient. A natural approach to obtain such algorithms is parallel computation. 
Motivated by such needs, the massively parallel computation (MPC) model~\cite{karloff2010model, andonimodel,Goodrichmodel,beamemodel} 
has been introduced to understand the power and limitations of parallel algorithms. It first proposed by Karloff et al.~\cite{karloff2010model} as a theoretical model to embrace Map Reduce algorithms, a class of powerful parallel algorithms not compatible with previously defined models for  parallel computation. Recent developments in the MPC model have made it a cornerstone for obtaining massively parallel algorithms. 

While in the previous parallel settings such as the PRAM model, usually an $O(\log n)$ factor in the round complexity is inevitable, MPC allows for sublogarithmic round complexity~\cite{karloff2010model,czumaj2017round,lattanzi2011filtering}.  Karloff et al.~\cite{karloff2010model}  also compared this model to PRAM, and showed that for a large portion of PRAM algorithms, there exists an MPC algorithm with the same number of rounds. In this model, each machine has unlimited access to its memory, however, two machines can only interact in between two rounds. Thus, a central parameter in this setting is the round complexity of algorithm since network communication is the typical main bottleneck in practice. The ultimate goal is developing constant-round algorithms, which are highly desirable in practice.

\vspace{0.2cm}
\noindent \textbf{The MPC model:} 
In this paper, we assume that the input size is bounded by $O(n)$, and we have $\machines$ machines of each with a memory of $\memory$. In the MPC model~\cite{karloff2010model, andonimodel,Goodrichmodel,beamemodel}, we assume the number of machines and the local memory size on each machine is asymptotically smaller than
the input size. Therefore, we fix an $0 < x < 1$ and bound the memory of each machine by $\tildorder(n^{1-x})$. Also, our goal is to have near linear total memory and therefore we bound the number of machines by $\tildorder(n^x)$.
An MPC algorithm runs in a number of rounds. In every round, every machine makes some local computation on its data. No communication between machines is allowed during a
round. Between two rounds, machines are allowed to communicate so long as each machine receives no more
communication than its memory. Any data that is outputted from a machine must be computed locally from the data
residing on the machine and initially the input data is distributed across the machines.

In this work we give MPC algorithms for different variants of the pattern matching problem. For the regular string matching and also \chr{?} and \chr{+} wildcard problems, our algorithms are tight in terms of running time, memory per machine, and round complexity. Both \chr{?} and \chr{+} wildcard problems are reduced to fast convolution at the end, and make use of the fact that FFT could computed in $O(1)$ MPC rounds with near-linear total running time and total memory. Also, for the case of \chr{*} wildcard we present nontrivial MPC algorithms for two special cases that mostly tend to happen in practice. However, the round complexity of these two cases is $O(\log(n))$, and the general case problem is not addressed in this paper.


\subsection{Our Results and Techniques}
Throughout this paper, we denote the text by $T$ and the pattern by $P$. Also, we denote the set of characters by $\Sigma$.

We begin, as a warm-up, in Section \ref{nowildcard} by giving a simple MPC algorithm that solves string matching in 2 rounds. The basic idea behind our algorithm is to cleverly construct hash values for the substrings of the text and the pattern. In other words, we construct an MPC data structure that enables us to answer the following query in a single MPC round:

 \textit{Given indices $i$ and $j$ of the text, what is the hash value for the substring of the text starting from position $i$ and ending at position $j$?}
 
  Indeed, after the construction of such a data structure, one can solve the problem in a single round by making a single query for every position of the text. This gives us a linear time MPC algorithm that solves string matching in constant rounds.

\vspace{0.2cm}
{\noindent \textbf{Theorem} \ref{mpcsm} (restated). \textit{There exists an MPC algorithm that solves string matching in constant rounds. The total memory and the total running time of the algorithm are linear.\\}}

For the case of wildcard \chr{?}, the hashing algorithm is no longer useful. It is easy to see that since special \chr{?} characters can be matched with any character of the alphabet, no hashing strategy can identify the matches. However, a more sophisticated coding strategy enables us to find the occurrences of the pattern in the text. Assume for simplicity that $m = |\Sigma|$ is the size of the alphabet and we randomly assign a number $1 \leq \map_c \leq m$ to each character $c$ of the alphabet. Moreover, we assume that all the numbers are unique that is for two characters $c$ and $c'$ we have  $\map_c = \map_{c'}$ if and only if $c = c'$. Now, construct a vector $T^\qvec$ of size $2|T|$ such that  $T^\qvec_{2i-1} = \map_{T_i}$ and $T^\qvec_{2i} = 1/ \map_{T_i}$ for any $1 \leq i \leq |T|$. Also, we construct a vector $P^\qvec$ of size $2|P|$ similarly, expect that $P^\qvec_{2i-1} = P^\qvec_{2i} = 0$ if the $i$'th character of $P$ is \chr{?}. Let $\nz_P$ be the number of the normal characters (\chr{?} excluded) of the pattern. It follows from the construction of $T^\qvec$ and $P^\qvec$ that if $P$ matches with a position $i$ of the text, then we have:
$$T^\qvec[2i-1,2i+2(|P|-1)] . P^\qvec = \nz_P$$
where $T^\qvec[2i-1,2i+2(|P|-1)]$ is a sub-vector of $A$ only containing indices $2i-1$ through $2i+2(|P|-1)$. Moreover, it is showed in \cite{FP74} that the vice versa also holds. That is if $T^\qvec[2i-1,2i+2(|P|-1)] . P^\qvec = \nz_P$ for some $i$ then pattern $P$ matches position $i$. This reduces the problem of pattern matching to the computation of dot products which is known to admit a linear time solution using fast Fourier transform (FFT)~\cite{brigham1988fast}. 

\vspace{0.2cm}
{\noindent \textbf{Theorem} \ref{mpcfft} (restated). \textit{There exists an MPC algorithm that computes FFT in constant rounds. The total memory and the total running time of the algorithm are $\tildorder(n)$.\\}}

Corollary \ref{smrepmpc} gives us an efficient MPC algorithm for the wildcard setting. While the reduction from wildcard matching to FFT is a known technique~\cite{cole2002verifying} the fact that FFT is computable in $O(1)$ MPC rounds leads to efficient MPC algorithm for a plethora of problems.  FFT is used in various combinatorial problems such as knapsack~\cite{bateni2018fast},  3-sum~\cite{chan2015clustered}, subset-sum~\cite{koiliaris2017faster},  tree-sparsity~\cite{backurs2017better}, tree-separability~\cite{bateni2018fast}, necklace-alignment~\cite{bremner2014necklaces}, etc.

The case of \chr{+} and \chr{*} wildcards are technically more involved. In the case of repetition, special characters \chr{+} may appear in the pattern. These special characters allow for arbitrary repetitions of a single character. For instance, a pattern \str{a+bb+} matches all strings of the form $\mathsf{a}^x . \mathsf{b}^y$ such that $x \geq 1$ and $y \geq 2$. Indeed this case is more challenging as the pattern may match with substrings of the text with different lengths.

To tackle this problem, we first compress both the text and the pattern into two strings $T^\rle$ and $P^\rle$ using the run-length encoding method. In the compressed versions of the strings, we essentially avoid repetitions and simply write the numbers of repetitions after each character. For instance, a text ``\textsf{aabcccddad}" is compressed into ``\textsf{a\color{red}[2]\color{black}b\color{red}[1]\color{black}c\color{red}[3]\color{black}d\color{red}[2]\color{black}a\color{red}[1]\color{black}d\color{red}[1]\color{black}}" and a pattern ``\textsf{ab+ccc+}" is compressed into ``\textsf{a\color{red}[1]\color{black}b\color{red}[1+]\color{black}c\color{red}[3+]\color{black}}". With this technique, we are able to break the problem into two parts. The first subproblem only incorporates the characters which is basically the conventional string matching. The second subproblem only incorporates repetitions. More precisely, in the second subproblem, we are given a vector $A$ of $n$ integer numbers and a vector $B$ of $m$ entries in the form of either $i$ or $i+$. An entry of $i$ matches only with indices of $A$ with value $i$ but an entry of $i+$ matches with any index of $A$ with a value at least $i$. Since we already know how to solve string matching efficiently, in order to solve the repetition case, we need to find a solution for the latter subproblem.


To solve this subproblem, we use an algorithm due to Cole and Hariharan~\cite{cole2002verifying}. They showed the subset matching problem could be solved in near-linear time. The definition of the subset matching problem is as follows.

\vspace{0.2cm}
{\noindent \textbf{Problem} \ref{problem:subset} (restated). \textit{Given $T$, a vector of $n$ subsets of the alphabet $\Sigma$, and $P$, a vector of $m$ subsets in the same format as $T$, find all occurrences of $P$ in $T$. $P$ is occurred at position $i$ in $T$ if for every $1 \leq j \leq |P|$, $P_j \subseteq T_{i+j-1}$.\\}}

We can reduce our problem to an instance of the subset matching problem by replacing every $T_i$ with $\{1+, 2+, \ldots, T_i+\} \cup \{T_i\}$, and keep $S$ intact, i.e., replacing each $S_i$ with $\{ S_i \}$. This way, if there is a match, each $S_i$ has to be included in the respective $T_j$. There is an algorithm for this problem with $\tildorder(s)$ running time~\cite{cole2002verifying}, where $s$ shows the total size of all subsets in $T$ and $P$.  The running time is good enough for our algorithm to be near linear since $s$ is at most twice the number of characters in the original text and pattern. Each $T_i$ is the compressed version of $T_i$ consecutive repetitions of a same character in $T$. We also show that the subset matching problem could be implemented in a constant number of MPC rounds, and thus a constant-round MPC algorithm for string matching with \chr{+} wildcard is implied.

{\noindent \textbf{Theorem} \ref{mpcsmrepeat} (restated). \textit{
There exists an MPC algorithm that solves string matching with \chr{+} wildcard in constant rounds. The total memory and the total running time of the algorithm are $\tildorder(n)$.\\}}

Despite positive results for \chr{?} and \chr{+} wildcards, we do not know whether a poly-logarithmic round MPC algorithm exists for \chr{*} wildcard in the most general case or not. This wildcard character matches any substring of arbitrary size in the text. We can consider a pattern consisted of \chr{*} and $\Sigma$ characters as a sequence of subpatterns, maximal substrings not having any \chr{*}, with a \chr{*} between each consecutive pair. In the sequential settings, we can solve the problem by iterating over the subpatterns, and find the next matching position in the text for each one. If we successfully find a match for every subpattern in order, we end up with a substring of $T$ matching $P$. Otherwise, it is easy to observe that $T$ does not match the pattern $P$. To find the next matching position, we can perform a KMP algorithm on $T$ and all the subpatterns one at a time, and make a transition to the next subpattern whenever we find a match, which is still linear in $n + m$, as the total size of the subpatterns is limited by $m$. The algorithm is provided with more details in Observation \ref{ob:star}.

However, things are not as easy in the MPC model, because each subpattern can happen virtually anywhere in the text, and furthermore the number of subpatterns could be as large as $O(m)$. Thus, intuitively, it is impossible to know which subpatterns match each location of $T$ with a linear total memory since the total size of this data could be $\Omega(nm)$, and also we cannot transfer this data in poly-logarithmic number of MPC rounds. Nonetheless, we provide two examples of how we can overcome this restriction by adding a constraint on the input. In both the cases, a near-linear $O(\log(n))$-round MPC algorithm exists for solving \chr{*} wildcard problem.

\begin{enumerate}
\item When the whole pattern fits in a single machine, $m = O(n^{1-x})$, then the number of subpatterns is limited by $O(n^{1-x})$. Thus, each machine could access each subpattern, and finds out if an interval of subpatterns happens in its part of input. At the end, we merge these pieces of information using dynamic programming to obtain the result.

\vspace{0.2cm}
{\noindent \textbf{Theorem} \ref{thm:star_1} (restated). \textit{
Given strings $s \in \alf^n$ and $p \in \{\alf \cup \text{\chr{*}}\}^m$ for $m = O(n^{1 - x})$, 
there is an MPC algorithm to find the solve the string matching problem in $O(\log n)$ rounds using $O(n^x)$ machines.\\}}

\item When no subpattern is a prefix of another, then the number of different subpatterns matching each starting position is limited by $1$. Exploiting this property, we can find out which subpattern matches each starting position, if any. Next, we create a graph with positions in $T$ as vertices, and we put an edge from a position $i$ which matches a subpattern $P_k$, to the minimum position $j$ matching subpattern $P_{k+1}$, which is also located after $i + |P_k| - 1$. This way, \chr{*} wildcard string matching reduces to a graph connectivity problem; finding whether there is path from a position matching the first subpattern to a position matching the last subpattern. Therefore, we will have a $O(\log(n))$-round MPC algorithm for \chr{*} wildcard problem using the standard graph connectivity results in MPC~\cite{karloff2010model}.

\vspace{0.2cm}
{\noindent \textbf{Theorem} \ref{thm:star_2} (restated). \textit{
Given strings $s \in \alf^n$ and $p \in \{\alf \cup \text{\chr{*}}\}^m$ such that  subpatterns are not a prefix of each other, 
there is an MPC algorithm to find the solve the string matching problem in $O(\log n)$ rounds using $O(n^x)$ machines.\\}}

\end{enumerate}

\begin{table*}[t]
\begin{tabular}{ |p{2cm}|p{2cm}|p{2cm}||p{2cm}|p{2cm}|p{2cm}|  }
 \hline
 \multicolumn{2}{|l}{Problem} &
 \multicolumn{1}{l}{Theorem} &
 \multicolumn{1}{l}{Rounds} &
 \multicolumn{1}{l}{Total Runtime} & 
 \multicolumn{1}{l|}{Nodes} \\
 \hline
 \multicolumn{2}{|l|}{$P \in \alf^m$} & \ref{mpcsm} &
  $O(1)$ & $O(n)$ & $O(n^x)$\\ \hline
 \multicolumn{2}{|l|}{$P \in \{\alf \cup \text{\chr{?}} \}^m$} & \ref{smrepmpc} &
  $O(1)$ & $\tildorder(n)$ & $O(n^x)$\\ \hline
 \multicolumn{2}{|l|}{$P \in \{\alf \cup \text{\chr{+}} \}^m$} &  \ref{mpcsmrepeat} &
 $O(1)$ & $\tildorder(n)$ & $O(n^x)$\\ \hline
 \multirow{2}{*}{$P \in \{\alf \cup \text{\chr{*}} \}^m$} & 
 $m = O(n^{1-x})$ &  \ref{thm:star_1} &
 $O(\log(n))$ & $O(n)$ & $O(n^{x})$\\ \cline{2-5}
 & no prefix &  \ref{thm:star_2} & 
 $O(\log(n))$ & $O(n)$ & $O(n^{x})$\\
  \hline
 \end{tabular}
\caption{Overview of results}
 \end{table*}
\section{Preliminaries}

In the pattern matching problem, we have two strings $T$ and $P$ of length $n$ and $m$ over an alphabet $\alf$. In the string matching problem, the first string $T$ is a text, and the second string $P$ is the pattern we are looking for in $T$. For a string $s$, we denote by $s[l,r]$ the substring of $s$ from $l$ to $r$, i.e., $s[l, r] = \langle s_l, s_{l+1}, \ldots, s_{r}\rangle$. Given two strings $T$ and $P$, we are looking for all occurrences of the pattern $P$ as a substring of $T$. In other words, we are looking for all $i \in [1, n - m + 1]$ such that $T[i,i+m-1] = P$. In this case, we say substring $T[i, i+m-1]$ matches pattern $P$.

\begin{problem}
Given two strings $T \in \alf^n$ and $P \in \alf^m$, we want to find all occurrences of the string $P$ (pattern) as a substring of string $T$ (text).
\label{smatch}
\end{problem}

Problem \ref{smatch} has been vastly studied in the literature, and there are a number of solutions that solve the problem in linear time. \cite{knuth1977fast} However, the main focus of this paper is to design massively parallel algorithms for the pattern matching problem. We consider MPC as our parallel computation model, as it is a general framework capturing state-of-the-art parallel computing frameworks such as Hadoop MapReduce and Apache Spark. 

We extend the problem of string matching adding wildcard characters to the pattern. A wildcard character is a special character $\phi \not\in \alf$ in pattern that is not required to match by the same character in $T$. For example, wildcard character \chr{?} can be matched by any arbitrary character in $T$. For instance, let $T$ and $P$ be \str{abracadabra} and \str{a?a} respectively. Then, pattern $P$ occurs at $i = 4$ and $i = 6$ since $P \matchq T[4,6]$ and $P \matchq T[6,8]$. Notation $\matchx{\phi}$ denoted the equality of two strings regarding the wildcard character $\phi$.

We consider three kinds of wildcard characters in this paper:

\begin{enumerate}
\item Character replacement wildcard `?': Any character can match character replacement wildcard.
\item Character repetition wildcard `+': If character repeat wildcard appears immediately after a character $c$ in the pattern, then any number of repetition for character $c$ is accepted.
\item String replacement wildcard `*': A string replacement wildcard can be matched with any string of arbitrary length.
\end{enumerate}

\section{String matching without wildcard}\label{nowildcard}

In this section, we provide an algorithm for the string matching problem with no wildcards; Given strings $T$ and $P$, we wish to find all the starting points in $T$ that match $P$. The round complexity of our algorithm, which is the most important factor in MPC model, is constant. Furthermore, the algorithm is tight in a sense that the total running time and memory is linear. For the more restricted cases, we enhance our algorithm to work in even less rounds.

The algorithm consists of three stages. We discuss the different stages of the algorithm with details in the following. The main goal of the algorithm is to compare the hash of each length $m$ substring of string $T$ with the hash of string $P$, and therefore, find an occurrence if they are equal. We consider the following properties, namely partially decomposability, for our hash function $h$,

\begin{itemize}
\item Merging the hashes of two strings $s$ and $s'$, to find the hash of their concatenation $s + s'$, can be done in $O(1)$.
\item Querying the hash of a substring $s[l,r]$ can be done in $O(1)$, with $O(|s|)$ preprocessing time. 
\end{itemize}

Further details regarding a hash function with partially decomposability property are provided in Appendix~\ref{hashing}.

Consider strings $T$ and $P$ are partitioned into several blocks of size $\memory$, each fit into a single machine. We denote these blocks by $T^1, T^2, \ldots, T^{n/\memory}$ and $P^1, P^2, \ldots, P^{m/\memory}$. It is easy to solve the problem when $P$ is small relative to $\memory$. We can search for $P$ in each block independently by maintaining the partial hash for each prefix of the block. The caveat of this solution is that some substrings lie in two machines. We can fix this issue by feeding the initial string chunks with overlap. If the length of the overlaps is greater than $m$, it is guaranteed than each candidate substring is contained as a whole in one of the initial string chunks. We call this method ``Double Covering''. Utilizing this method, we can propose an algorithm which leads to Observation \ref{mpcdoublecovering}. Therefore, the main catch of the algorithm is to deal with the matching when string $P$ spans multiple blocks.

\begin{observation} Given $T \in \alf^n$ and $P \in \alf^m$ where $m = O(n^{1-x})$, there is an MPC algorithm for string matching problem with no wildcards in $1$ MPC round using $O(n^x)$ machines in linear total running time.
\label{mpcdoublecovering}
\end{observation}

\begin{proof}
We provide for each machine, which holds $T^i$, the next $m$ characters in the $T$, i.e., $T[i\memory+1:i\memory+m]$, which is viable since $m = O(\memory)$. Therefore, we can compute the hash of each substring starting in $T^i$ as well as $P$ in each machine. See Algorithm \ref{alg:nw1} for more details.
\end{proof}

\begin{algorithm}[h!]
\KwData{
two array $T$ and $P$, where $m = O(\memory)$}
\KwResult{
function $f : \{1,2,\ldots, n-m+1 \} \to \{ 0, 1\}$ such that $f(i) = 1$ iff $T[i:i+m-1] = P$. }

$f(i) \gets 0 \beforall 1 \leq i \leq n-m+1$\;

send a copy of $P$ as well as $T[i\memory+1:i\memory + m]$ to the $i$-th machine, which holds $T^i$, truncate if $i\memory + m$ surpasses $n$.

Run in parallel: \For{$1 \leq i \leq n/\memory$}{
	append $T[i\memory+1:i\memory+m]$ to $T^i$\;
	\For{$1 \leq j \leq \memory$}{
		\If{$\hash(P) = \hash(T^i[j:j+m-1])$}{
			$f((i-1)\memory + j) \gets 1$\;
		}
	}
}
\caption{\textsf{StringMatching}(a)}
\label{alg:nw1}
\end{algorithm}

\begin{theorem} For any $x < 1/2$, given $T \in \alf^n$ and $P \in \alf^m$, there is an MPC algorithm for string matching problem with no wildcards in $O(1)$ MPC rounds using $O(n^x)$ machines in linear total running time.
\label{mpcsm}
\end{theorem}

Intuitively, we can handle larger patterns by transferring the partial hash of each ending point to the machine which the respective starting point lies in. In addition, we  compute the hash of each $k$ first blocks, and by which, we can find the hash of each interval of blocks. Note that we need to assume $x < 1/2$, so that memory of a single machine be capable of storing $O(1)$ information regarding each machine. The algorithm is outlined in Algorithm \ref{alg:nw2}.

\begin{proof}
Assume there are $(n+m)/\memory+2$ machines as following:

\begin{itemize}
\item machines $a_1, a_2, \ldots, a_{n/\memory}$ which $T^i$ is stored in $a_i$.
\item machines $b_1, b_2, \ldots, b_{m/\memory}$ which $P^i$ is stored in $b_i$.
\item machines $c$ and $d$ which are used for aggregation purposes.
\end{itemize}

At the first round, we compute the hash of each block $T^i$ and $P^i$ and send them to machines $c$ and $d$ respectively. Furthermore, we compute the partial hashes of each prefix of each block $T^i$, that is $\hash(T^i[1,j])$ for $1 \leq j \leq \memory$. The calculated hash of each prefix should be sent to the machine containing the respective starting point. Therefore, $T^i[1,j]$ should be sent to the node containing starting point $(i-1)\memory+j-m+1$. The total communication overhead of this round is $O(n+m)$.

In the next round, we aggregate the calculated hashes in the previous round. In node $d$, we calculate the hash of string $P$ by merging the hashes of $T^1, T^2, \ldots, T^{m/\memory}$. The value of $\hash(P)$ should be sent to all $a_i$ for $1 \leq i \leq n/\memory$. Those nodes are supposed to compare $\hash(P)$ with the hash of each starting point lying in their block to find the matches in the last round. Furthermore, in node $c$, the hash of each first $k$ blocks of $T$ should be calculated, i.e., $\hash(T^1+T^2+\ldots+T^i)$ for all $1 \leq i \leq n/\memory$. Each machine $a_i$ needs to compute the hash of a sequence of consecutive blocks of $T$ that lie between each starting point in $T^i$ and its respective end point. Therefore, each machine should receive up to $O(1)$ of these hashes, because the set of respective ending points spans at most $2$ blocks. 

In the final round, we have all the data required for finding the matches with starting point inside block $T^i$ at node $a_i$. It suffices to find the hash of each suffix of $T^i$, i.e., $\hash(T^i[l, \memory])$, in which $l$ is the candidate starting point. Consider $T^k$ is the block where the ending point respective to $l$ lies inside, and $r$ is the index of the ending point inside $T^k$. Then, using the hash of the sequence of block between $T^i$ and $T^k$, and $\hash(T^k[1,r])$ (both the hashes has been sent to $a_i$ in the previous steps) one can decide whether the starting point $l$ in block $T^i$ is a match or not.
\end{proof}

\begin{algorithm}[h!]
\KwData{
two array $T$ and $P$}
\KwResult{
function $f : \{1,2,\ldots, n-m+1 \} \to \{ 0, 1\}$ such that $f(i) = 1$ iff $T[i:i+m-1] = P$. }

$f(i) \gets 0 \beforall 1 \leq i \leq n-m+1$\;

Run in parallel: \For{$1 \leq i \leq m/\memory$}{
	send $\hash(P^i)$ to machine $d$\;
}

Run in parallel: \For{$1 \leq i \leq n/\memory$}{
	send $\hash(T^i)$ to machine $c$\;
	\For{$1 \leq j \leq \memory$}{
		send $\hash(T^i[1,j])$ to the machine $a_k$, where the corresponding starting point, i.e., $(i-1)\memory+j-m+1$, lies in $T^k$\;
	}
}

send $\hash(P) \gets \textsf{merge}(\hash(P^1), \hash(P^2), \ldots, \hash(P^{m/\memory}))$ to each $a_i$ for $1 \leq i \leq n/\memory$\;

\For{$1 \leq i \leq n/\memory$}{
	$\hash(T^1+T^2+\ldots+T^i) \gets \textsf{merge}(\hash(T^1+T^2+\ldots+T^{i-1}), \hash(T^i))$\;
	send $\hash(T^1+T^2+\ldots+T^i)$ to each machine $a_j$ that needs it in the next round\;
}

Run in parallel: \For{$1 \leq i \leq n/\memory$}{
	\For{$1 \leq l \leq \memory$}{
		$s \gets (i-1)\memory+l$\;
		let $k$ be the index of the machine where $s+m-1$ lies in $T^k$\;
		$r \gets s+m-1-(k-1)\memory$\;
		$\eta = \hash(T^{i+1}+T^{i+2}+\ldots+T^{k-1})$\;
		$\hash(T[s, s+m-1]) \gets \textsf{merge}(\hash(T^i[l, \memory]), \eta, \hash(T^k[1, r]))$\;
		\If{$\hash(T[s, s+m-1]) = \hash(P)$}{
			$f(s) \gets 1$\; 
		}
	}
}

\caption{\textsf{StringMatching}(b)}
\label{alg:nw2}
\end{algorithm}

\section{Character Replace Wildcard \chr{?}}\label{charreplace}

We start this section summarizing the algorithm for string matching with \chr{?} in sequential settings, which require Fast Fourier Transform. 

\begin{theorem}[Proven in \cite{FP74}]
Given strings $T \in \alf^n$ and $P \in (\alf \cup \{?\})^m$, it is possible to find all the substrings of $T$ that match with pattern $P$ in $\widetilde{O}(n+m)$.
\label{smrep}
\end{theorem}

The idea behind Theorem \ref{smrep} is taking advantage of fast multiplication algorithms, e.g. using Fast Fourier Transform, to find the occurrences of $P$ which contains `?' wildcard characters. Fischer and Paterson \cite{FP74} proposed the first algorithm for string matching with `?' wildcard. The main idea is to replace each character $c \in \alf$ in string $T$ or $P$ with two consecutive numbers $\map_c$ and $1/\map_c$ and each `?' with two consecutive zeros. If we compute the convolution of $T$ and the reverse of $P$, we can ensure a match if the convoluted value with respect to a substring of $T$ equals the number non-wildcard characters in $P$, aka $\nz_P$. This procedure requires a non-negligible precision in float arithmetic operations that adds a $\log(|\alf|)$ factor to the order of the algorithm.  In the subsequent works, the dependencies on the size of alphabet has been eliminated.~\cite{kalai2002efficient,indyk1998faster,CC07}

As stated in Theorem \ref{smrep}, we can solve string matching with \chr{?} wildcards in $O(n\log(n))$ using convolution. Convolution can be computed in $O(n\log(n))$ by applying Fast Fourier Transform on both the arrays, performing a point-wise product, and then applying inverse FFT on the result. Therefore, we should implement FFT in a constant round MPC algorithm in order to solve pattern matching with wildcard \chr{?}. Inverse FFT is also possible with the similar approach.
 
Given an array $A = \langle a_0, a_1, \ldots, a_{n-1} \rangle$, we want to find the Discrete Fourier Transform of array $A$. Without loss of generality, we can assume $n = 2^k$, for some $k$, to avoid the unnecessary complication of prime factor FFT algorithms; it is possible to right-pad by zeroes otherwise. By Fast Fourier Transform, applying radix-2 Cooley-Tukey algorithm for example, one can compute the Discrete Fourier Transform of $A$ in $O(n\log(n))$ time, which is defined as:

\begin{align}
a^*_k = \sum_{j=0}^{n-1}{a_j\cdot \eunity{n}{j k}} \label{dftO}
\end{align}

Where $A^* = \langle a^*_0, a^*_1, \ldots, a^*_{n-1} \rangle$ is the DFT of $A$. We interchangeably use the alternative notation $\unity{n}{} = \eunity{n}{}$, by which Equality \ref{dftO} becomes

\begin{align}
a^*_k = \sum_{j=0}^{n-1}{a_j\cdot \unity{n}{jk}} \label{dft}
\end{align}

We state the following theorem regarding the FFT in the MPC model.

\begin{theorem}
For any $x \leq 1/2$, a collection of $O(n^x)$ machines each with a memory of size $O(n^{1-x})$ can solve the problem of finding the Discrete Fourier Transform of $A$ with $O(1)$ number of rounds in MPC model. The total running time equals $O(n\log(n))$, which is tight.
\label{mpcfft}
\end{theorem}

Roughly speaking, our algorithm is an adaptation of Cooley-Tukey algorithm in the MPC model. In Appendix~\ref{subsection:mpcfft}, we justify Theorem \ref{mpcfft} by giving an overview of the algorithm, and then we show how it results in a $O(1)$-round algorithm for computing the DFT of an array in the MPC model.

\begin{corollary}
Given strings $T \in \alf^n$ and $P \in (\alf \cup \{?\})^m$, it is possible to find all the substrings of $T$ that match with pattern $P$ with a $O(1)$-round MPC algorithm with total runtime and memory of $\widetilde{O}(n+m)$.
\label{smrepmpc}
\end{corollary}

\begin{proof}
Build vectors $T^\qvec$ and $P^\qvec$ as following:

\begin{align*}
T^\qvec = \langle \map_{T_1}, \map_{T_1}^{-1}, \map_{T_2}, \map_{T_2}^{-1}, \ldots, \map_{T_n}, \map_{T_n}^{-1} \rangle\\
P^\qvec = \langle \map_{P_1}, \map_{P_1}^{-1}, \map_{P_2}, \map_{P_2}^{-1}, \ldots, \map_{P_n}, \map_{P_n}^{-1} \rangle
\end{align*}

Note that $\map_c$ for all characters $c \in \alf$ has a positive integer value, and $\map_c^{-1}$ equals $1/\map_c$. However, let $\map_? = \map_?^{-1} = 0$ for wildcard character \chr{?}. For example, for text \str{abracadabra} and pattern \str{a?a}, considering $mp_c$ equals the index of each letter in the English alphabet, we will have

\begin{align*}
& T^\qvec = \langle 
1, \frac{1}{1},
2, \frac{1}{2},
18, \frac{1}{18},
1, \frac{1}{1},
3, \frac{1}{3},
1, \frac{1}{1},
4, \frac{1}{4},
1, \frac{1}{1},
2, \frac{1}{2},
18, \frac{1}{18},
1, \frac{1}{1} \rangle\\
& P^\qvec = \langle
1, \frac{1}{1},
0, 0,
1, \frac{1}{1} \rangle
\end{align*}

Let $\conv = T^\qvec \circledast \textsf{rev}(P^\qvec)$, where operator $\circledast$ shows the convolution of its two operands. We can observe that for all $1 \leq i \leq n-m+1$ we have 

\begin{align}
\conv_{2i+2m-1} = \sum_{j=1}^{2m}{T^\qvec_{2(i-1)+j} \cdot P^\qvec_{j}}
\end{align}

It is clear that if $T[i, i+m-1]$ matches pattern $P$, then $\conv_{2i+2m-1} = 2\nz_P$ since each the wildcard characters add up to $0$, and for any other character $c$, $\map_c \times 1/\map_c + \map_c \times 1/\map_c = 2$. According to \cite{FP74}, the other side of this expression also holds, i.e., $T[i, i+m-1]$ matches pattern $P$ if $\conv_{2i+2m-1} = 2\nz_P$. Therefore, to find that whether a substring $T[i, i+m-1]$ matches pattern $P$ or not, it suffices to

\begin{enumerate}
\item Compute $\conv = T^\qvec \circledast \textsf{rev}(P^\qvec)$ utilizing FFT, which can be implemented in $O(1)$ rounds in MPC model as stated in Theorem \ref{mpcfft}.
\item For all $1 \leq i \leq n-m+1$, substring $T[i, i+m-1]$ matches pattern $P$ iff $\conv_{2i+2m-1} = 2\nz_P$.
\end{enumerate}

\end{proof}

\section{Character Repetition Wildcard \chr{+}}\label{charrepeat}

The character repetition wildcard, shown by \chr{+}, allows expressing ``an arbitrary number of a single character" within the pattern. An occurrence of \chr{+} which is immediately after a character $c$, matches arbitrary repetition of character $c$. For example, text \str{bookkeeper} has a substring that matches pattern \str{oo+k+ee+}, but none of its substrings match pattern \str{oo+kee+}. In Subsection \ref{ss:gtm}, we reduce pattern matching with wildcard \chr{+} to greater-than matching in $O(1)$ MPC rounds. The greater-than matching problem is defined with further details in Problem \ref{problem:gtm}. Afterwards, we show a reduction of greater-than matching from subset matching problem, explained in \ref{problem:subset}. Then, by showing subset matching could be implemented in $O(1)$ MPC rounds, we propose an $O(1)$-round MPC algorithm for string matching with wildcard \chr{+} in Theorem \ref{mpcsmrepeat}.

\subsection{Relation to greater-than matching}\label{ss:gtm}

\begin{problem}
\emph{Greater-than matching:} Given two arrays $T$ and $P$, of length $n$ and $m$ respectively, find all indices $1 \leq i \leq n-m+1$ such that the continuous subarray of $T$ starting from $i$ and of length $m$ is element-wise greater than $P$, i.e., $T_{i+j-1} \geq P_{j} \beforall 1 \leq j \leq m$.
\label{problem:gtm}
\end{problem}

To reduce pattern matching with wildcard \chr{+} to greater-than matching, we perform run-length encoding \ref{def:rle} on both the text and the pattern. This way, we have a sequence of letters each along with a number showing the number of its repetitions, or a lower-bound restricting the number of repetitions in case we have a \chr{+} wildcard. Subsequently, we can solve the problem for letters and numbers separately, and merge the result afterwards. Note that pattern $P$ matches a substring of $T$ if and only if the both the respective letters and the respective repetition restrictions match. We also show we can find the run-length encoding of a string in $O(1)$ MPC rounds in Observation \ref{obs:rlempc}.

\newcommand{\cnt}{{\textsf{cnt}}}

\begin{definition}
For an arbitrary string $s$, let $s^\rle$ be the run-length encoding of $s$, computed in the following way:

\begin{itemize}
\item Ignoring \chr{+} characters, decompose string $s$ into maximal blocks consisting of the same character representing by pairs $\langle c_i, \cnt_i \rangle$ which show a block of $\cnt_i$ repetitions of character $c_i$. 
\item If a \chr{+} character is located immediately after a block or within a block, that block becomes a wildcard block, and represented as $\langle c_i, \cnt_i+ \rangle$, where $\cnt_i$ is still the number of the occurrences of character $c_i$ in the block.
\item $s^\rle$ equals the list of these pairs $\langle c_i, \cnt_i \rangle$ or $\langle c_i, \cnt_i+ \rangle$, concatenated in a way that preserves the original ordering of the string.
\end{itemize}
\label{def:rle}
\end{definition}

For example, for $T = \textsf{\str{bookkeeper}}$ and $P = \textsf{\str{o+o+k+ee+p}}$,
\begin{align*}
& T^\rle = \langle \langle \textsf{b}, 1 \rangle,
 \langle \textsf{o}, 2 \rangle,
 \langle \textsf{k}, 2 \rangle,
 \langle \textsf{e}, 2 \rangle,
 \langle \textsf{p}, 1 \rangle,
 \langle \textsf{e}, 1 \rangle,
 \langle \textsf{r}, 1 \rangle \rangle \\ 
& P^\rle = \langle \langle \textsf{o}, 2+ \rangle,
 \langle \textsf{k}, 1+ \rangle,
 \langle \textsf{e}, 2+ \rangle,
 \langle \textsf{p}, 1 \rangle \rangle
\end{align*} 
 
 We alternatively show the compressed string as a string, for example, 

\begin{itemize}
\item
$T^\rle = $
 \str{\textsf{b\color{red}[1]\color{black}o\color{red}[2]\color{black}k\color{red}[2]\color{black}e\color{red}[2]\color{black}p\color{red}[1]\color{black}e\color{red}[1]\color{black}r\color{red}[1]\color{black}}}.
\item
$P^\rle = $
 \str{\textsf{o\color{red}[2+]\color{black}k\color{red}[1+]\color{black}e\color{red}[2+]\color{black}p\color{red}[1]\color{black}}}.
\end{itemize}

\begin{observation}\label{obs:rlempc}
	We can perform run-length encoding in $O(1)$ MPC rounds.
\end{observation}

\begin{proof}
	Suppose string $s$ is partitioned among $\machines$ machines such that $i$'th machine contains $s_i$ for $1 \leq i \leq \machines$. Run-length encoding of each $s_i$ could be computed separately inside each machine. Let 
	$$ s_i^\rle = \langle \langle c_{i,1}, cnt_{i,1} \rangle, \langle c_{i,2}, cnt_{i,2} \rangle, \ldots, \langle c_{i,l_i}, cnt_{i,l_i} \rangle \rangle$$ 
	be the run-length encoding of $s_i$, where $l_i$ is its length. Then in the next round, we can merge $\langle c_{i,l_i}, cnt_{i, l_i} \rangle$ and $\langle c_{i+1,1}, cnt_{i+1,1} \rangle$ if $c_{i,l_i} = c_{i+1,1}$ for all $1 \leq i \leq \machines - 1$, and we will end up with $s^\rle$ if we concatenate all $s_i^\rle$s.
\end{proof}

As we mentioned before, in order to reduce from greater-than matching, it is possible to divide the problem into two parts: matching the letters, and ensuring whether the repetition constraints are hold, and solve each part separately. The former is a simple string matching problem, but the latter requires could be solved with greater-than matching. Formally, we need to find all indices $1 \leq i \leq n - m + 1$ such that for every $1 \leq j \leq m$, the following constraints holds:

\begin{enumerate}
\item $T^\rle_{i+j-1} = \langle c_j, x \rangle$ for some $x \geq \cnt_j$ if $P^\rle_j = \langle c_j, \cnt_j+ \rangle$.
\label{constraint:1}
\item $T^\rle_{i+j-1} = \langle c_j, \cnt_j \rangle$ if $P^\rle_j = \langle c_j, \cnt_j \rangle$ and $2 \leq j \leq m - 1$.
\label{constraint:2}
\item $T^\rle_{i+j-1} = \langle c_j, x \rangle$ for some $x \geq \cnt_j$ if $P^\rle_j = \langle c_j, \cnt_j \rangle$ and $j \in \{1, m\}$.
\label{constraint:3}
\end{enumerate}

Constraint \ref{constraint:3} needs to be considered to allow the substring in the original text $T$ starts and ends in the middle of a block of $c_1$ or $c_m$ letters. By considering only those substrings which their letters match, we can get rid of letter constraints. Also to ensure constraint \ref{constraint:2}, we can perform a wildcard \chr{?} pattern matching by replacing each $\cnt_j+$ and also $\cnt_1$ and $\cnt_m$ by a \chr{?} wildcard, and keep only the numbers that constraint \ref{constraint:2} checks. Thus, if a substring $T^\rle_{i, i+m-1}$ matches according to this wildcard \chr{?} pattern matching, what remains is a greater-than matching, as we only need to check constraints \ref{constraint:1} and \ref{constraint:3}, and we can replace numbers we already checked for constraint \ref{constraint:2} by some $0$'s. 

\begin{observation}
We can reduce pattern matching with wildcard \chr{+} from greater-than matching in $O(1)$ MPC rounds.
\label{obs:mpcgtmreduce}
\end{observation}

Using Observation \ref{obs:rlempc}, it only remains to perform $O(1)$ wildcard \chr{?} matchings to obtain an instance of greater-than matching that already satisfies constraint \ref{constraint:2}, as well as letter constraints. Thus, Observation \ref{obs:mpcgtmreduce} is implied.

\subsection{Reduction from subset matching}\label{ss:subset}

\begin{problem}
\emph{Subset Matching:} Given $T$, a vector of $n$ subsets of the alphabet $\Sigma$, and $P$, a vector of $m$ subsets in the same format as $T$, find all occurrences of $P$ in $T$. $P$ is occurred at position $i$ in $T$ if for every $1 \leq j \leq |P|$, $P_j \subseteq T_{i+j-1}$.\label{problem:subset}  \end{problem}

The subset matching problem is a variation of pattern matching where pattern matches text if each of the pattern entries is a subset of the respective entries in text. Keeping this in mind, we use subset matching to solve greater-than matching in $O(1)$ MPC rounds, and thereby achieving a $O(1)$-round solution for the problem of pattern matching with \chr{+} wildcard character using the subset matching algorithm proposed by Cole and Hariharan~\cite{cole2002verifying}.

\begin{observation}\label{lemma:mpcsubsetreduce}
We can reduce greater-than matching to subset matching in $O(1)$ MPC rounds with total runtime and total memory of $O(Q)$, where $Q$ is the sum of all $T_t$'s, i.e. $Q = \sum_{i = 1}^{m}{T_i}$.
\end{observation}

The idea behind Observation \ref{lemma:mpcsubsetreduce} is to have a subset $\{ 0, 1, 2, \ldots, T_i\}$ instead of each entry of $T_i$, and a subset $\{ P_i \}$ instead of each $P_i$. This way if $T$ matches $P$ in a position $1 \leq i \leq n - m +1$, then we have $P_j \in \{ 0, 1, 2, \ldots, T_{i+j-1} \}$ since $P_j \leq T_{i+j-1}$ for all $1 \leq j \leq m$. This way, we can solve pattern matching with wildcard \chr{+} in $O(1)$ MPC rounds if subset matching could be solved in $O(1)$ MPC rounds. In the following, Lemma \ref{lemma:mpcsubset} shows it is possible to solve subset matching w.h.p in a constant number of MPC rounds. The algorithm utilizes $O(log(n))$ invocations of wildcard \chr{?} matching.

\begin{lemma}\label{lemma:mpcsubset}
	We can solve subset matching in $O(1)$ MPC rounds with total runtime and total memory of $\tildorder(Q)$, where $Q$ is the sum of the size all $T_i$'s, i.e. $Q = \sum_{i = 1}^{n}{|T_i|}$.
\end{lemma}

\newcommand{\dec}{{\mathcal T}}
\newcommand{\decp}{{\mathcal P}}

\begin{proof}
	First, consider solving an instance of the greater than matching problem in $O(1)$ MPC rounds.
	We partition the entries inside all $T^i$'s into $\Sigma$ sequences $T^{i,1}, T^{i,2}, \ldots, T^{i,\Sigma}$ inside each machine, where $T^{i,j} = \{k \mid j \in T_{iS + k}\}$. We just create a subset of these sets which are not empty, i.e., $\dec_i = \{T^{i,j} \mid |T^{i,j}| > 0\}$, so that $\dec_i$ fits inside the memory of each machine. $\decp_i$ also can be defined similarly. We can sort all the union of all $\dec_i$'s and $\decp_i$'s in a non-decreasing order of $j$ and breaking ties using $i$ in $O(1)$ MPC rounds~\cite{goodrich2011sorting}. 
	
	This way, we end up with a sparse wildcard matching for each character in $\Sigma$, because we want to check in which substrings of $T$ each occurrence of character $j \in \Sigma$ in $P$ is contained in the corresponding $T_i$. We can put a \chr{1} instead of each occurrence of $j$ in $P$,  a \chr{?} instead of each $j$ in $T$, and a \chr{0} in all other places since we are considering a greater than matching instance. Using the similar algorithm as in~\cite{cole2002verifying}, we can easily reduce each instance of sparse wildcard matching to $O(\log(k))$ normal wildcard matching with size of $O(k)$ in $O(1)$ MPC rounds, where $k$ is the number of non-zero entries. Using Corollary \ref{smrepmpc}, we can give a $O(1)$ round MPC algorithm for subset matching when the input is an instance of greater than matching. We can easily extend these techniques to solve subset matching algorithm in $O(1)$ MPC rounds, as Cole and Hariharan ~\cite{cole2002verifying} showed it is analogous to sparse wildcard matching.
\end{proof}

\begin{theorem}
There exists an MPC algorithm that solves string matching with \chr{+} wildcard in constant rounds. The total memory and the total running time of the algorithm are $\tildorder(n)$.
\label{mpcsmrepeat}
\end{theorem}

\begin{proof}
We can also simplify the algorithm for pattern matching with wildcard \chr{+}, by exploiting subset matching flexibility. We can also use subset matching to verify constraint \ref{constraint:2} of greater-than matching reduction (instead of wildcard \chr{?} pattern matching), as well as letter constraint (instead of regular string matching with no wildcard). We define $T'$ and $P'$ as follows:

\begin{align}
T'_i &= \bigcup_{j=1}^{cnt_i}{\{ \langle c_i, j+ \rangle\}} \cup \{ \langle c_i, cnt_i \rangle\} 
& \textsf{if } T^\rle_i = \langle c_i, cnt_i \rangle \\
P'_i &= \{ \langle c_i, cnt_i \rangle\} 
& \textsf{if } P^\rle_i = \langle c_i, cnt_i \rangle\\
P'_i &= \{ \langle c_i, cnt_i+ \rangle\} 
& \textsf{if } P^\rle_i = \langle c_i, cnt_i+ \rangle
\end{align}

It could easily be observed that subset matching of $T'$ and $P'$ is equivalent to pattern matching with wildcard \chr{+} of $T$ and $P$, and also this simpler reduction is straight-forward to implement in $O(1)$ MPC rounds. In addition, the total runtime and memory of this subset matching which is $\tildorder(Q )$ is equal to $\tildorder(n)$ in the original input. Note that $T'$ is resulted form $T^\rle$ whose sum of its numbers, that each of them shows the repetitions of the respective letter, equals $n$.
\end{proof}

\section{String Replace Wildcard `*'}\label{strreplace}
In this section we consider the string matching problem with wildcard `*'. 

Given strings $s \in \alf^n$ and $p \in \{\alf \cup \text{\chr{*}}\}^m$, 
we say a substring of $p$ is a subpattern if it is a maximal substring not containing `*'.
We present the following results in this section:
\begin{enumerate}
\item A sequential algorithm for string matching with wildcard `*' in time $O(n+m)$.
\item An MPC algorithm  for string matching with wildcard `*' in $O(\log n)$ rounds using $O(n^x)$ machines if the length of pattern $p$ is at most $O(n^{1-x})$.
\item An MPC algorithm  for string matching with wildcard `*' in $O(\log n)$ rounds using $O(n^x)$ machines if all the subpatterns of $p$ are not prefix of each other.
\end{enumerate}

\subsection{Linear time sequential algorithm}
\begin{observation}\label{ob:star}
Given strings $s \in \alf^n$ and $p \in \{\alf \cup \text{\chr{*}}\}^m$, 
there is a sequential algorithm to decide if $s$  matches with pattern $p$ in time $\tilde{O}(n+m)$.
\end{observation}
Let the subpatterns of $p$ to be $P_1, P_2, \dots, P_w$. Our sequential algorithm is \textsf{StringMatchingWithStar}($a$).

\begin{algorithm}[h!]
\KwData{
$s \in \alf^n$ and $p \in \{\alf \cup \text{\chr{*}}\}^m$.}
\KwResult{
Yes or No. 
}
Set $i \leftarrow 1$\;

	\For {$j = 1, 2, \dots, w$ and $i \leq n$}{
		Run KMP for the string $s[i, n]$ and pattern $P_j$\;
		If KMP fails, then Return No\;
		Let $i'$ be the position satisfying $s[i', i' + |P_j| - 1] = P_j$ obtained by KMP\;
		Set $i \leftarrow i' + |P_j|$\;
	}
Return Yes.

\caption{\textsf{StringMatchingWithStar}($a$)}
\label{alg:ewdcs}
\end{algorithm}

\subsection{MPC algorithm for small subpattern}

\begin{theorem}\label{thm:star_1}
Given strings $s \in \alf^n$ and $p \in \{\alf \cup \text{\chr{*}}\}^m$ for $m = O(n^{1 - x})$, 
there is an MPC algorithm to find the solve the string matching problem in $O(\log n)$ rounds using $O(n^x)$ machines. 
\end{theorem}

We assume string $s$ is partitioned into $s_1, s_2, \dots, s_t$ for some $t = O(n^x)$ such that every $s_i$ has length at most $O(n^{1-x})$.
We say string $s$ is an exact matching of $p$ if there is a partition of $s$ into $|p|$ (possibly empty) substrings such that if $p[i]$ is not `*', then $i$-th substring is same to $p[i]$.

Given  indices $i, j \in [t]$ of string $s$ and position $k$ of pattern $p$,
let $f(i, j, k)$ be the largest position of $p$ such that the concatenation of $s_i, s_{i+1}, \dots, s_j$ matches pattern $p[k, f(i, j, k)]$. 

To illustrate our idea, we need the following definitions:
\begin{enumerate}
	\item $g(k)$ for position $k$ of $p$: the largest position which is smaller than or equal to $k$ such that $p[k]$ is `*'.
	\item $h(k)$ for position $k$ of $p$ such that $p[k] \neq$ `*': the smallest integer $r$ such that $p[g(k) + 1, k - r]$ is a prefix of $p[g(k) + 1, k]$.
\end{enumerate}

Consider the following equation for an arbitrary $i \leq i' < j$ (let $\beta = f(i, i', k)$)

\begin{itemize}
\item
if $p[\beta] = \text{\chr{*}}$:
$$ f(i, j, k) = f(i'+1, j, \beta) $$
\item
if $p[\beta] \neq \text{\chr{*}}$:

\end{itemize}
\begin{equation}\label{equ:logrounds}f(i, j, k) =
\max
\begin{cases}
f(i'+1, j, g(\beta)), \\ 
\max_{0 \leq \ell \leq \lfloor \frac{\beta - h(\beta)}{g(\beta)}\rfloor }\left\{f(i'+1, j, \beta- \ell\cdot h(\beta ))\right\}
\end{cases}
\end{equation}
Our algorithm is \textsf{StringMatchingWithStar}(b). 

\begin{proof}[of \ref{thm:star_1}]
We show that Eq.~\ref{equ:logrounds} correctly computes $f(i, j, k)$ for an arbitrary $i \leq i' < j$. Then the theorem follows from the description of \textsf{StringMatchingWithStar}(b),  Eq.~\ref{equ:logrounds}  and Observation~\ref{ob:star}. 

Let $\alpha$ be the largest position of $p$ such that 
there is an exact matching of the concatenation of $s_i, \dots, s_j$ and $p[k, \alpha]$. 
If $f(i, i', \cdot)$ and $f(i'+1, j, \cdot)$ are correct, and $f(i, j, k)$ is computed by Eq.~\ref{equ:logrounds}, then $f(i, j, k) \leq \alpha$, since  Eq.~\ref{equ:logrounds} implies a feasible exact matching.

Now we show that $f(i, j, k) \geq \alpha$ if  $f(i, j, k)$ is computed by Eq.~\ref{equ:logrounds}.
There is an exact matching of the concatenation of $s_i, \dots, s_j$ and $p[k, \alpha]$. 
Let $\gamma$ be the position of pattern $p$ such that the last symbol of $s_{i'}$ is matched to $p[\gamma]$.
By the definition of function $f$, $\gamma \leq \beta$.

Consider the case of $p[\gamma]  = $`*' or $p[\gamma]  \neq $`*' but $p[\gamma]  = $`*'.
We have $g(\beta) \geq \gamma$ and $f(i'+1, j, \gamma) = \alpha$.
By the monotone property of $f$, we have $f(i, j, k) \geq f(i'+1, j, g(\beta)) \geq f(i'+1, j, \gamma)  = \alpha$.

Consider the case of $p[\gamma] \neq $`*' and $p[\beta] \neq $`*'.
If $\gamma$ and $\beta$ are in different subpatterns, then using above argument, we have $f(i, j, k) \geq  \alpha$.
Otherwise, $p[h(\beta) + 1, \gamma]$ is a suffix of $p[h(\beta)+1, \beta]$, which implies that there is a non-negative integer $\ell$ such that $\gamma = \beta - \ell \cdot h(\beta)$.
Hence, $f(i, j, k) \geq  \alpha$.
\end{proof}

\begin{algorithm}[h!]
\KwData{
two array $s$ and $p$.}
\KwResult{
Yes or No. }

Distribute $s_1, s_2, \dots, s_t$ to distinct machines, and distribute $p$ to every machine\;

Compute $f(i, i, k)$ for all the $k$ on the machine containing $s_i$ by algorithm  \textsf{StringMatchingWithStar}($a$) \textbf{in parallel}\;

\For{$m = 1, 2, \dots,  \lceil \log_2 t \rceil$}{
	For every $i \geq 1$,  put $f(i, i+2^{m - 1} - 1, k)$ and $f(i+2^{m-1}, i+2^m - 1, k)$ into same machine for all the $k$  \textbf{in parallel}\;
	For every $i \geq 1$, compute $f(i, i+2^m - 1, k)$ for all the $k$ by Equation~\ref{equ:logrounds} with $j = i+2^m - 1$ and $i' = i+2^{m-1} - 1$ \textbf{in parallel}\;

}
Return Yes if $f(1, t, 1) = w$, otherwise return No. 
\caption{\textsf{StringMatchingWithStar}(b)}
\label{alg:ewdcs}
\end{algorithm}
\subsection{MPC algorithm for non-prefix subpatterns}

\begin{theorem}\label{thm:star_2}
	Given strings $s \in \alf^n$ and $p \in \{\alf \cup \text{\chr{*}}\}^m$ such that  subpatterns are not a prefix of each other, 
	there is an MPC algorithm to find the solve the string matching problem in $O(\log n)$ rounds using $O(n^x)$ machines. 
\end{theorem}

\begin{proof}
	We show that algorithm \textsf{StringMatchingWithStar}(c) solves the problem. 
	
	We first prove the correctness of the algorithm. 
	Since all the subpatterns are not a prefix of each other, 
	for every position $i$ of string $s$, there is at most one subpattern $P_u$ such that $P_u$ is a prefix of $s[i, n]$.
	On the other hand, if $s[i, j]$ is not a prefix of any subpattern, then $h(s[i, j])$ does not equal to any of the hash value obtained in Step 1, 
	otherwise, $h(s[i, j])$ is equal to some hash value obtained in Step 1. 
	Hence, for every $i \in[n]$,  the ``for" loop of Step 2 finds the subpattern $P_u$ such that $P_u$ is a prefix of $s[i, n]$ by binary search if $P_u$ exists. 
	
	If string $s$ matches pattern $p$,  then any set of positions $a_1, a_2, \dots, a_u$ with the following two conditions 
	\begin{enumerate}
		\item $P_i$ is a prefix of $s[a_i, n]$ for every $i \in [w]$.
		\item $a_i + |P_i| \leq a_{i+1}$ for every $i \in [w-1]$.
	\end{enumerate}
	corresponds to a matching between $s$ and $p$.
	Hence, string $s$ matches pattern $p$ if and only if $v_0$ and $v_{m+1}$ are connected in  the constructed graph of Step 15 and 16.
	
	Now we consider the number of MPC rounds required. 
	Using the argument of Section~\ref{nowildcard}, computing hash of all the prefixes of every subpattern or a set of $n$ substrings of $s$  needs constant MPC rounds. 
	Hence Step 1 needs constant rounds, and  the ``for" loop of Step 2 needs $O(\log n)$ rounds. 
	Step 15 naturally needs a single round.
	Step 16 can be done in $O(\log n)$ rounds by sorting all the $(i, f(i))$ pairs according to the $f$ function values and selecting $j$ such that $j \geq i + |P_{f(i)}|$ and $f(j) = f(i) + 1$ for all the pairs $(i, f(i))$.
	The graph connectivity needs $O(\log n)$ rounds. 
\end{proof}

\begin{algorithm}[h!]
	\KwData{
		two array $s$ and $p$.}
	\KwResult{
		Yes or No. }
	
	
	For each subpattern $P_i$ \textbf{in parallel } compute the hash value of every prefix of $p_i$\;

	\For {each position $i$ of string $s$ \text{\textbf{in parallel }} } 
	{
		Set $j = 0$ and $k = n$ initially\;
		\While {$j < k$} {
			Let $\ell = \lceil(j+k)/2\rceil$\;
			Compute the hash $h(s[i, i+ \ell])$\;
			\If { there is a hash obtained in step 2 same to $h(s[i, i+\ell])$}  {
				Set $j \leftarrow \ell$\;
			}
			\Else {
				Set $k \leftarrow \ell - 1$\;
			}
		}
		\If {there is a subpattern $p_u$ same to $s[i, i+j]$} 
		{
			Set $f(i) \leftarrow u$\; 
		}
		\Else {
			Set $f(i) \rightarrow 0$\;
		}
	}
	
	Construct an empty graph \textbf{in parallel} with vertices $v_0, v_1, \dots, v_{n+1}$. Add edge $(v_0, v_s)$ and $(v_t, v_{n+1})$ where $s$ is the smallest integer such that $f(s) = 1$,
	$t$ is the largest integer such that $f(t) = w$\;
	
	For every $i \in [n]$ such that $f(i) > 0$, \textbf{in parallel} 
	add edge $(v_i, v_j)$ to the graph where $j$ is the smallest integer such that $j \geq i + |P_{f(i)}|$ and $f(j) = f(i) + 1$\;
	
	Run graph connectivity algorithm on the graph constructed. return Yes if $v_0$ and $v_{n+1}$ are in the same connected component of the graph, otherwise return No \;

	\caption{\textsf{StringMatchingWithStar}(c)}
	\label{alg:ewdcs}
\end{algorithm}

\section{Acknowledgements}
	MohammadTaghi Hajiaghayi and Hamed Saleh are supported in part by NSF BIGDATA grant IIS-1546108, and NSF SPX grant CCF-1822738.  Saeed Seddighin is partially funded by an unrestricted gift from Google and a Research
	Award from Adobe.

\bibliographystyle{abbrvnat}
\bibliography{stringmatching}

\appendix
\newpage

\section{Hashing}\label{hashing}
To boost comparing the pattern $p$ with the substrings of string $s$, we can compare the fingerprints or hashes instead. A hash of string $s$, denoted by $h(s)$, is a single number such that $h(s) \not= h(s')$ implies $s \not= s'$ always, and $h(s) = h(s')$ implies $s = s'$ with a high probability. Consider a naive hash function $h_r$ such that $h_r(s) = \sum_{i=1}^{|s|}{s_i} \mod r$, where $r$ is an arbitrary number. Note that we can label the characters in $\alf$ with a permutation of size $|\alf|$, thereby, considering a numerical value for each character to simplify the formulas. This hash function does not guarantee a high probability of $s = s'$ in the case of $h_r(s) = h_r(s')$, because this hash function is not locally sensitive, i.e., modifying two consecutive characters in $s$ in a way that the first one reduced by one and the other one increased by one gives us the same hash. 

However, simple hash function $h_r$ has a desirable property that is rolling. A rolling hash function $h$ allows us to achieve $h(s[l+1, r+1])$ from $h(s[l,r])$ in $O(1)$ time, i.e., $h_r(s[l+1,r+1]) = h_r(s[l,r]) - s_l + s_{r+1} \mod r$. Using a rolling hash function, one can imagine a Monte-Carlo randomized algorithm for string matching problem by first computing $h(p)$ and $h(s_{1,m})$ in $O(m)$ time, and then building $h(s[i+1, i+m])$ from $h(s[i, i+m-1])$ for $i$ from $1$ to $n-m$ in $O(n)$ time. Afterwards, we can compare the hash of each substring with $h(p)$ in $O(1)$ time. Therefore, we desire a locally sensitive rolling hash function. Karp and Robbin \cite{KR87} provided such rolling hash functions that guarantee a bounded probability of hash collision. An example of such hash function is 

$$ h_{r,b}(s) = \sum_{i=1}^{|s|}{s_i b^{|s|-i}} \mod r $$

The difference between this hash function and the previous one is local sensitivity. It's not likely anymore to end up with the same hash applying few modifications if we chose $r$ to be a prime number or simply chose such $r$ and $b$ to be relatively prime numbers. Besides, the hash collision probability of an ideal hash function is $r^{-1}$, and no matter how we minimize the correlation of the hashes of similar strings, we cannot achieve a smaller probability. By the way, $h_{r,b}$ with relatively prime $r$ and $b$ achieve a probability in that this correlation is almost negligible. Therefore, it suffices to pick a large enough $r$ to guarantee high probability. More precisely, we want $\bigcup_{i=1}^{n-m+1}{\Pr[P\not=T[i,i+m-1] \wedge h(P)=h(T[i,i+m-1])]}$ be at most $1/n$. Thus,

\begin{align*}
\bigcup_{i=1}^{n-m+1}{\Pr[P\not=T[i,i+m-1] \wedge h(P)=h(T[i,i+m-1])]} & \leq \frac{1}{n} \Rightarrow \\
\sum_{i=1}^{n-m+1}{\Pr[P\not=T[i,i+m-1] \wedge h(P)=h(T[i,i+m-1])]} & \leq \frac{1}{n} \Rightarrow \\
O(n) \frac{1}{r} \leq \frac{1}{n} \Rightarrow r \geq O(n^2) \Rightarrow \log(r) \geq & O(\log(n))
\end{align*}

Therefore, the number of bits required to store hashes should be logarithmic in the size of the text which is a fairly reasonable assumption. Alternatively, we can reduce the probability substantially by comparing hashes based on multiple values of $r$. Hence, we can safely assume that using locally such sensitive hash functions, we can guarantee a high probability of avoiding hash collisions.

\begin{fact}
Given two strings $T \in \alf^n$ and $P \in \alf^m$, using locally sensitive rolling hash functions one can find all occurrences of the string $P$ (pattern) as a substring of string $S$ (text) in $O(n+m)$ time with a high probability.
\label{smhash}
\end{fact}

In addition to rolling property, we can consider a more general property for $h_{r,b}$ namely partially decomposability. A hash function $h$ is partially decomposable if it is possible to calculate  in $O(1)$ time $h(s[l,r])$ from the hash of the prefixes of $s$, that is $h(s[1,i])$ for all $i \in [1,|s|]$.  In addition, it should be possible to calculate the hash of the concatenation of two strings in $O(1)$ time. For example, $h_{r,b}(s[l,r]) = (h_{r,b}(s[1,r]) - h_{r,b}(s[1,l-1]) b^{r-l+1}) \mod r$, and $h_{r,b}(s + s') = (h_{r,b}(s) b^{|s'|} + h_{r,b}(s')) \mod r$, where $s + s'$ means the concatenation of $s$ and $s'$.

Throughout this paper, we refer to the suitable hash function for an instance of string matching problem by $h$ and ignore the internal complications of the hash functions.

\section{$O(1)$-round algorithm for FFT}
\label{subsection:mpcfft}

To find the $O(1)$ round algorithm in MPC model, we first review how one can find the Discrete Fourier Transform in $O(n\log(n))$ time. The following recurrence helps us to solve the problem using a divide and conquer algorithm. \cite{danielson1942some} Let $\Psi[2k] = \langle a_0, a_2, \ldots, a_{n-2} \rangle$, i.e., an array of the even-indexed elements of $A$, and $\Psi[2k+1] = \langle a_1, a_3, \ldots, a_{n-1} \rangle$, an array of the odd-indexed elements of $A$ likewise. If $\Psi[2k]^* = \langle \psi[2k]^*_0, \psi[2k]^*_1, \ldots, \psi[2k]^*_{n/2-1} \rangle$ and $\Psi[2k+1]^* = \langle \psi[2k+1]^*_0, \psi[2k+1]^*_1, \ldots, \psi[2k+1]^*_{n/2-1} \rangle$ show the DFT of $\Psi[2k]$ and $\Psi[2k+1]$ respectively, then for all $0 \leq j < n / 2$

\begin{align}
	\left.
	\begin{cases}
		a^*_{j} & = \psi[2k]^*_{j} + \unity{n}{j} \psi[2k+1]^*_{j}  \\
		a^*_{j+n/2} & = \psi[2k]^*_{j} - \unity{n}{j} \psi[2k+1]^*_{j} \\
	\end{cases}
	\right.
	\label{recur}
\end{align}

Intuitively, the recurrence decomposes $A$ into two smaller arrays of size $n/2$, $\Psi[2k]$ and $\Psi[2k+1]$, and represents the DFT of $A$ based on $\Psi[2k]^*$ and $\Psi[2k+1]^*$. Accordingly, we can solve two smaller DFT problems, each with size $n/2$, and then merge them in $O(n)$ time to find $A^*$. Therefore, we can find $A^*$ in $O(n\log(n))$ time since the depth of the recurrence is $O(\log(n))$. We cannot find a constant round MPC algorithm exclusively using this recurrence. To extend the recursion, we can define $\Psi[pk+q]$ as the following,

$$ \Psi[pk+q] = \langle a_{q}, a_{p+q}, \ldots,  a_{(l-1)p + q} \rangle \hspace{0.5cm} \forall \,\, 0 \leq q < p \leq n$$

Where $l = \lfloor \frac{n-1-q}{p}\rfloor+1$ is the number of the indices whose reminder is $q$ in division by $p$. $\Psi[pk+q]$ contains all the elements of $A$ whose index is $q$ modulo $p$. Note that the notation of $\Psi[pk+q]$ assumes a universal quantification on all $0 \leq k \leq n/p$. We also show the DFT of $\Psi[pk+q]$ by 

$$ \Psi[pk+q]^* = \langle \psi[pk+q]^*_0, \psi[pk+q]^*_1, \ldots,  \psi[pk+q]^*_{l-1}  \rangle $$

We can observe that for all $0 \leq b \leq \log_2(n)$ and $0 \leq q < 2^{b+1}$, assuming $p = 2^b$, the following counterpart of Recurrence \ref{recur} holds for all $0 \leq j < l$, where $l$ is the length of $\Psi[2^bk+q]$ which equals $\log(n)-b$:

\begin{align}
	\left.
	\begin{cases}
		\psi[2^bk+q]^*_{j} & = \psi[2^{b+1}k+q]^*_{j} + \unity{l}{j}\psi[2^{b+1}k+2^b+q]^*_{j}  \\
		\psi[2^bk+q]^*_{j+l/2} & = \psi[2^{b+1}k+q]^*_{j} - \unity{l}{j}\psi[2^{b+1}k+2^b+q]^*_{j}  \\
	\end{cases}
	\right.
	\label{nsrecurgen}
\end{align}

Since $\Psi[2^{b+1}k+q]$ and $\Psi[2^{b+1}k+2^b+q]$ decompose $\Psi[2^bk+q]$ into the even-indexed and odd-indexed elements, Recurrence \ref{nsrecurgen} is implied by plugging in $\Psi[2^bk+q]$ as $A$ in Recurrence \ref{recur}. The following Lemmas (\ref{depend} and \ref{bitrev}) point out two properties regarding FFT in the MPC model. These properties immediately give us an $O(\log(n))$-round algorithm for the FFT problem.

\begin{lemma}
	For all $\log(\machines) \leq b \leq \log(n)$ and $0 \leq q < 2^{b}$, the value of $\Psi[2^{b}k+q]^*$ could be computed in a single machine with memory of $\memory$.
	\label{depend}
\end{lemma}

\begin{proof}
	If all of the entries of $\Psi[2^{b}k+q]$ exist in the memory of a single machine, which is viable since $b \geq \log(\machines)$, then we can compute $\Psi[2^{b}k+q]^*$ as following:
	
	$$\Psi[2^{b}k+q]^* = \textsf{fft}(\Psi[2^{b}k+q])$$
\end{proof}

In many implementations of Fast Fourier Transform, bit-reversal technique is used to facilitate the algorithm from various perspectives, to achieve an in-place FFT algorithm for example. Bit-reversal technique reorders the initial array based on the reverse binary representation of the indices; Considering $(10100)_2$ as the sort key for $a_5$, $5 = (00101)_2$, when $n = 32$ for example. In other words, bit reversal technique permutes the elements of $A$ by a permutation $\pi$ such that $\pi(i) = \textsf{rev}(i)$, where $\textsf{rev}(i)$ is the reverse binary presentation of $i$.  Applying bit reversal technique guarantees the locality of data, especially in the first $\log(\memory)$ stages. 

\begin{lemma}
	The bit-reversal operation makes a permutation of input entries which if we it split into $2^{b}$ continuous chunks of equal size, then for all $0 \leq b \leq \log(n)$, $\Psi[2^b k+\textsf{rev}(q)]$ would the $q$-th chunk.
	\label{bitrev}
\end{lemma}

\begin{proof}
	Notice that the bit-reversal permutation is actually sorting the entries based on the reverse binary representation of the index of each entry. Therefore, all indices with the same low $b$ bits form a continuous interval of size $2^{\log(n) - b}$ in the bit-reversal permutation. Thus, $\Psi[2^b k + \textsf{rev}(q)]$ which all of its members have the same low $b$ bits, $\textsf{rev}(q)$. Moreover, it is the $q$-th chunk in the bit-reversal permutation, because $\textsf{rev}(q)$ ranks $q$-th in the bit-reversal ordering of all non-negative integers less than $2^b$.
\end{proof}

\begin{definition}
	For all $0 \leq q < \machines$, Let $\Phi_q = \Psi[\machines k + \textsf{rev}(q)]$. ($\Phi_q$ shows the input of the $q$-th machine after performing the bit-reversal permutation, according to Lemma \ref{bitrev}) We also show the DFT of $\Phi_q$ by $\Phi^+_q$, instead of $\Phi^*_q$. The $j$-th element of $\Phi_q$ ($\Phi^+_q$) is denoted by $\phi_{q,j}$. ($\phi^+_{q,j}$)
	\label{phi}
\end{definition}

Utilizing Lemma \ref{depend} and Lemma \ref{bitrev}, we can achieve a $O(\log(n))$-round MPC algorithm which merges all $\Phi_q^+$s in $\log(\machines)$ steps. At the first step, we apply bit-reversal permutation in order to gather each $\Phi_q$ in a single machine. Applying bit-reversal permutation does not have any special communication overhead, as we are just transferring the elements. Then we compute $\Phi^+_q$ for all $q \in \mathbb{Z}_m$. Afterwards, in the $l$-th step, for $0 \leq l < \log(m)$, we can merge $2^{\log(m)-l}$ blocks of size $2^{\log(\memory)+l}$ using equation \ref{nsrecurgen} in pairs.

It seems that $\log(n)$ MPC rounds is required because we need to merge $\Phi^+_q$s which are distributed in various machines, and according to Equality \ref{dftO}, each $a^*_k$ depends on  every $\Phi^+_q$ for $q \in \mathbb{Z}_m$. However, we can exploit the nice properties of the graph of dependencies of $a^*_i$s to $\Phi^+_q$s, which is also called the \textit{Butterfly Graph}, to decompose these dependencies efficiently. An example of this dependencies graph is demonstrated in Figure \ref{fftgraph_mpc}.  We can concentrate these dependencies in single machines by gathering $\phi^+_{q,j}$s with the same $j$. Figure \ref{fftgraph_mpc} demonstrates why it suffices to have $\phi^+_{q,j}$s with the same $j$ stored in a single machine. In this figure, $\memory$ and $\machines$ are equal to $4$, and the entries that should be stored in the same machine are colored with the same color. At the first step, the entries are distributed in bit-reversal permutation, and continuous chunks are stored in each machine. However, in the next step, we can see for example $a^*_5$, which is colored green, depends only on the green entries, which all have the same $j = 1$.

\begin{definition}
	Let $\widetilde{\phi}^+_{q,j} = \unity{n}{qj} \phi^+_{q, j}$. These coefficients are called twiddle factors, which allows us to express $a^*_k$ based on the FFT of $\widetilde{\phi}^+_{q,j}$s. Let $\Phi^*_j$ be the DFT of vector $\langle 
	\widetilde{\phi}^+_{0, j},
	\widetilde{\phi}^+_{1, j}, \ldots,
	\widetilde{\phi}^+_{\machines-1, j}
	\rangle$. We will show in Lemma \ref{butterfly} that  $a^*_{q\memory + j}$ equals $\phi^*_{j, q}$, which is the $j$-th element of $\Phi^*_j$.
\end{definition}

\begin{figure}[!h]
	\begin{center}
		\includegraphics[scale=0.1]{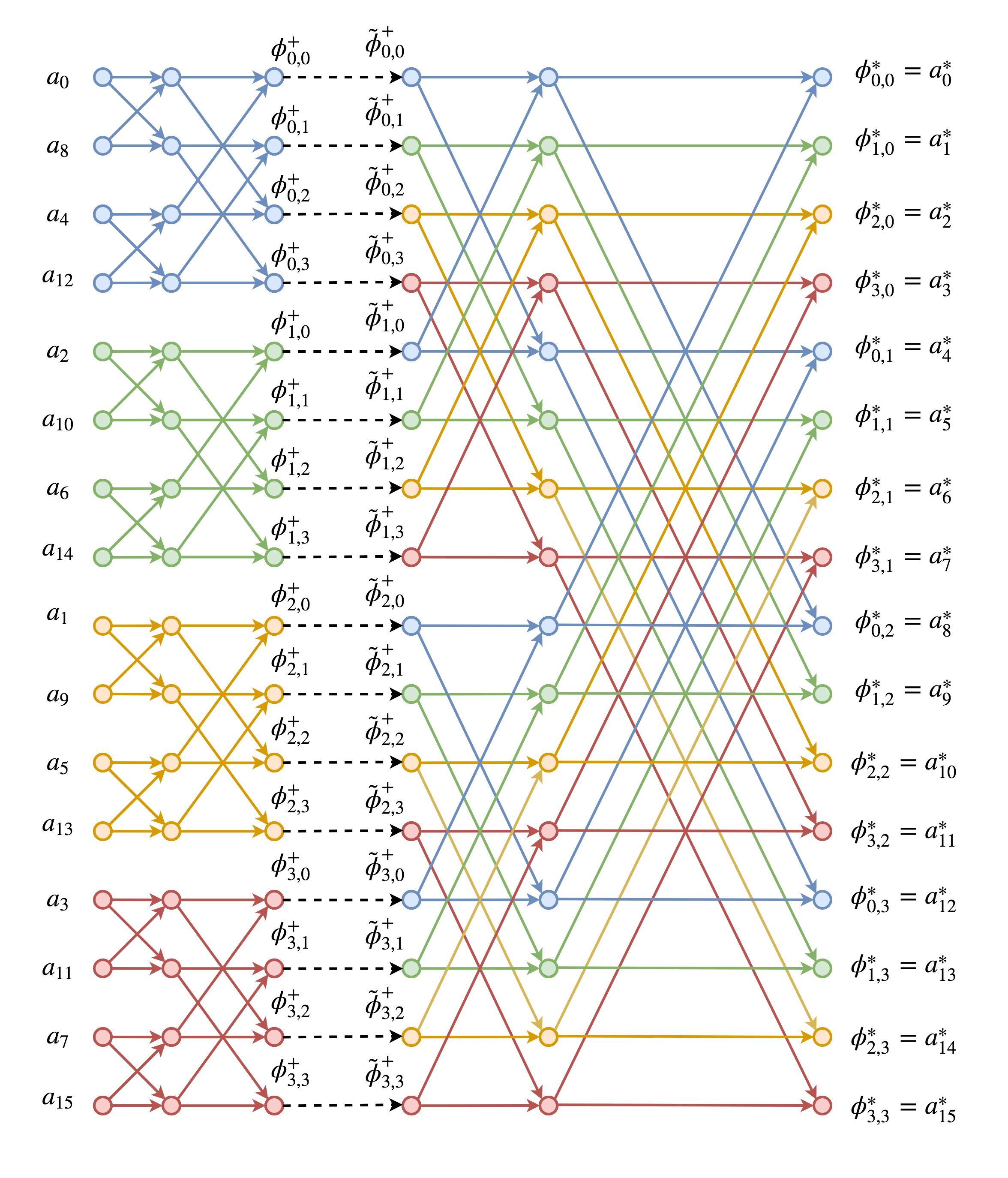}
	\end{center}
	\caption{The dependency graph of the FFT recurrence in MPC model. Different colors represent different machines, in a setting with $4$ machines with memory of $4$ each. It could be observed that the dependency graph in each machine at the first stage is isomorphic to the one of the second stage.}
	\label{fftgraph_mpc}
\end{figure}

\begin{lemma}
	For all $0 \leq j \leq \memory - 1$ and $0 \leq q \leq \machines -1 $, we have $a^*_{q\memory+j} = \phi^*_{j,q}$.
	\label{butterfly}
\end{lemma}

\begin{proof}
	We know that $a^*_k$ equals a weighted sum of all $a_j$ where weights are the $n$-th roots of unity to the power of $k$. i.e.,
	
	\begin{align*}
		a^*_k = \sum_{j=0}^{n-1}{\unity{n}{jk}a_j}
	\end{align*}
	
	We break down the summation into $\machines$ smaller summations of size $\memory$ such that each smaller summation represent one $\Phi_i$ for some $0 \leq i \leq \machines - 1$. Moreover, we show $k$ by $q\memory+j$, where $0 \leq j \leq \memory - 1$. Hence, we can show the following.
	
	\begin{align}
		a^*_{q\memory+j} & = \sum_{z = 0}^{\machines - 1}{
			\sum_{l=0}^{\memory - 1}{
				\unity{n}{(l\machines+z)(q\memory+j)} a_{l\machines + z}
			}
		}\\
		& = \sum_{z = 0}^{\machines - 1}{
			\sum_{l=0}^{\memory - 1}{
				\unity{n}{l\machines q\memory} \unity{n}{l\machines j} \unity{n}{zq \memory} \unity{n}{zj} a_{l\machines + z}
			}
		}\\
		& = \sum_{z = 0}^{\machines - 1}{
			\sum_{l=0}^{\memory - 1}{
				\unity{\memory}{l j} \unity{\machines}{zq} \unity{n}{zj} a_{l\machines + z}
			}
		}\label{simpWs}\\
		& = \sum_{z = 0}^{\machines - 1}{
			\unity{\machines}{zq} \unity{n}{zj} \sum_{l=0}^{\memory - 1}{
				\unity{\memory}{l j} a_{l\machines + z}
			}
		}\\
		& = \sum_{z = 0}^{\machines - 1}{
			\unity{\machines}{zq} \unity{n}{zj} \phi^+_{z,j}
		}\\
		& = \sum_{z = 0}^{\machines - 1}{
			\unity{\machines}{zq} \widetilde{\phi}^+_{z,j}
		} = \phi^*_{j,q}
	\end{align}
	
	Equality \ref{simpWs} is implied from the fact that $\unity{kn}{k} = \unity{n}{}$, and therefore, $\unity{n}{l\machines p \memory} = \unity{n}{lpn} = \unity{1}{lp} = 1$, $\unity{n}{l\machines j} = \unity{\memory}{l j}$, and $\unity{n}{zp\memory} = \unity{\machines}{zp}$.
\end{proof}

	Utilizing Lemma \ref{butterfly}, we can provide an algorithm which solves the FFT problem in $O(1)$ MPC rounds, because we already know how to compute $\phi^*_{j,q}$s. Algorithm \ref{alg:mpc_fft} shows the algorithm in details. It can be observed that our algorithm is analogous to Cooley-Tukey algorithm with radix $n^x$.
	
	\begin{algorithm}[h!]
		\KwData{
			an array $a$.}
		\KwResult{
			array $a^*$ containing the DFT of $a$.}
		
		permute $a$ such that members of $\Phi_q$ be together in a single machine for all $q \in \mathbb{Z}_\machines$.
		
		Run in parallel: \For{$q \in \mathbb{Z}_\machines$}{
			$\Phi^+_q \gets $ \textsf{fft}($\Phi_q$)\;
			$\widetilde{\phi}^+_{q,j} \gets \unity{n}{qj} \phi^+_{q,j} \beforall j \in \mathbb{Z}_\memory$\;
		}
		
		distribute $\widetilde{\phi}^+_{q,j}$ into different machines such that entries with same $j$ be in a same machine.
		
		Run in parallel: \For{$j \in \mathbb{Z}_\memory$}{
			$\Phi^*_j \gets $ \textsf{fft}($\langle \widetilde{\phi}^+_{0,j},
			\widetilde{\phi}^+_{1,j}, \ldots, \widetilde{\phi}^+_{\machines-1,j}  \rangle$)\;
		}
		
		permute all $\phi^*_{j,q}$ ($= a^*_{q\memory+j}$) in a way that each of them retain its correct position.
		\caption{\textsf{mpc-fft}($b$)}
		\label{alg:mpc_fft}
	\end{algorithm}
	
	\begin{proof}[Theorem \ref{mpcfft}]
		At the first step of the algorithm, we apply the bit-reversal permutation in the input, thereby grouping the members of $\Phi_q$ in each of $\machines$ machines. Applying a permutation imposes a communication overhead of $O(\memory)$ for each machine. Now, the $q$-th machine can compute $\Phi^+_q$ and $\widetilde{\Phi}^+_j$ standalone in $\widetilde{O}(n)$ time as following:
		\begin{itemize}
			\item $\Phi^+_q = \textsf{fft}(\Phi_q)$
			\item $\widetilde{\Phi}^+_{q,j} = \unity{n}{qj}\Phi^+_{q,j}$
		\end{itemize}
		
		In the next round, we distribute $\widetilde{\Phi}^+_{q,j}$s with the same $j$ in the same machines. Note that $\widetilde{\Phi}^+_{q,j}$s with the same $j$ fit into a single machine since $x \leq 1/2$. We even might need to fit multiple groups of $\widetilde{\Phi}^+_{q,j}$s into a single machine, which causes no problem. By the way, we can compute $\Phi^*_j$s in this stage in $\widetilde{O}(n)$ time. Afterwards it suffices restore appropriate $a^*_k$s to their original position, where for $k = qS+j$, $a^*_k = \phi^*_{j, q}$.
	\end{proof}

\begin{observation}
	\label{obs:ffttightrounds}
	Although Algorithm \ref{alg:mpc_fft} requires $3$ rounds of communication for computing DFT in the MPC model, the convolution of two arrays can be computed in $4$ communication rounds.
\end{observation}

\begin{proof}
	Since we need to apply FFT twice for computing the convolution, we can find the convolution in $6$ communication rounds. However, since the first round applies a permutation on the elements, and the last round applies the inverse permutation, we can ignore the last round of computing the FFT, along with the first round of computing the inverse FFT. The point-wise product operation which is needed two be done between two FFTs allows us to do so, as it is independent of the order of the arrays.
\end{proof}

\end{document}